\documentclass{article}
\usepackage{amsthm,amsmath,amssymb}
\usepackage[margin=.5in,bottom=1in, top=1in]{geometry} 
\usepackage{caption}
\usepackage{subcaption}
\usepackage{graphicx}
\usepackage{epstopdf}
\usepackage{color}
\usepackage{framed}
\usepackage{lscape}
\usepackage{rotating}
\usepackage{algorithm}
\usepackage[noend]{algpseudocode}
\usepackage{grffile}
\usepackage{multirow}
\usepackage{xcolor}
\usepackage{comment}
\usepackage{xr}
\usepackage[colorlinks=true,linkcolor=black,citecolor=black,urlcolor=black]{hyperref}

\newtheorem{lem}{Lemma}
\newtheorem{prop}{Proposition}
\newtheorem{cor}{Corollary}
\newtheorem{rmk}{Remark}
\newtheorem{thm}{Theorem}

\newcommand{\ep}{\varepsilon}

\newcommand{\R}{\mathbb{R}}
\newcommand{\A}{\mathcal{A}}

\newcommand{\la}{\langle}

\newcommand{\ra}{\rangle}

\newcommand{\diag}{\text{diag}}
\newcommand{\what}{\widehat}

\newcommand{\X}{\mathbf{X}}
\newcommand{\G}{\mathbf{G}}

\newcommand{\Y}{\mathbf{Y}}
\newcommand{\N}{\mathbf{N}}

\renewcommand{\u}{\mathbf{u}}
\renewcommand{\v}{\mathbf{v}}

\newcommand{\I}{\mathbf{I}}
\renewcommand{\A}{\mathbf{A}}

\newcommand{\B}{\mathbf{B}}

\newcommand{\s}{\underline{s}}

\newcommand{\umu}{\underline{\mu}}
\newcommand{\unu}{\underline{\nu}}

\newcommand{\bSigma}{\mathbf{\Sigma}}

\begin{document}
\title{Rapid evaluation of the spectral signal \\ detection threshold and Stieltjes transform 
%
}
\author{William Leeb\footnote{School of Mathematics, University of Minnesota, Twin Cities. Minneapolis, MN, USA.}}
\date{}
\maketitle

\abstract{
Accurate detection of signal components is a frequently-encountered challenge in statistical applications with low signal-to-noise ratio. This problem is particularly challenging in settings with heteroscedastic noise. In certain signal-plus-noise models of data, such as the classical spiked covariance model and its variants, there are closed formulas for the spectral signal detection threshold (the largest sample eigenvalue attributable solely to noise) for isotropic noise in the limit of infinitely large data matrices. However, more general noise models currently lack provably fast and accurate methods for numerically evaluating the threshold.

In this work, we introduce a rapid algorithm for evaluating the spectral signal detection threshold in the limit of infinitely large data matrices. We consider noise matrices with a separable variance profile (whose variance matrix is rank one), as these arise often in applications. The solution is based on nested applications of Newton's method. We also devise a new algorithm for evaluating the Stieltjes transform of the spectral distribution at real values exceeding the threshold. The Stieltjes transform on this domain is known to be a key quantity in parameter estimation for spectral denoising methods. The correctness of both algorithms is proven from a detailed analysis of the master equations characterizing the Stieltjes transform, and their performance is demonstrated in numerical experiments.
}

\section{Introduction}

Random matrix theory is an increasingly popular tool for deriving methods in statistical signal processing applications. Recent applications include signal denoising \cite{shabalin2013reconstruction},  \cite{gavish2017optimal},  \cite{leeb2020operator},  \cite{nadakuditi2014optshrink},  \cite{leeb2019optimal}, covariance estimation \cite{donoho2018optimal},  \cite{leeb2019optimal}, filter correction in signal and image processing \cite{bhamre2016denoising},  \cite{dobriban2020optimal}, multiple-input multiple-output (MIMO) wireless communication \cite{couillet2011random},  \cite{biglieri2007mimo}, and factor analysis \cite{dobriban2017deterministic},  \cite{dobriban2017factor}, to name just a few.

Part of the appeal of random matrix theory is that for many classes of random matrices, certain random quantities of statistical significance converge to well-defined limits as the matrix size grows infinitely large. By way of introduction, we illustrate this phenemonon in the \emph{spiked covariance model}, introduced by Johnstone in \cite{johnstone2001distribution}. Here, the user observes $l$ iid random vectors $Y_1,\dots,Y_l$ in $\R^k$ of the form $Y_j = X_j + \ep_j$, where $X_j$ (the signal component) is a random vector with covariance $\bSigma$ of rank $r \ll \min(k,l)$ and eigenvalues $\ell_1 > \dots > \ell_r > 0$, and $\ep_j$ (the noise component) is a zero mean Gaussian vector with covariance $\I_k$. Typically, the user wishes to estimate quantities related to the signal covariance $\bSigma$, such as its rank, eigenvectors, and eigenvalues, from the signal-plus-noise observations $Y_1,\dots,Y_l$.

It is well-known that when $k$ and $l$ are both large, there emerge simple relationships between the sample covariance $\what \bSigma = \frac{1}{l}\sum_{j=1}^{l} (Y_j - \overline Y) (Y_j - \overline Y)^T$ (where $\overline Y = \frac{1}{l}\sum_{j=1}^{l}Y_j$ denotes the sample mean) and the signal covariance $\bSigma$. For example, the largest $r$ eigenvalues of $\what \bSigma$ converge almost surely to the following limits as $k,l \to \infty$ and $k/l \to \gamma > 0$:
\begin{align}
\label{eq:spike}
\lambda_j = 
\begin{cases}
\left(\ell_j + 1 \right) \left(1 + \frac{\gamma}{\ell_j} \right), 
    & \text{ if } \ell_j > \sqrt{\gamma}  \\
(1 + \sqrt{\gamma})^2, & \text{ if } \ell_j \le \sqrt{\gamma}.
\end{cases}.
\end{align}
Furthermore, the inner products $\langle \u_j, \hat \u_m \rangle$ between the top $r$ eigenvectors $\u_1,\dots,\u_r$ of $\bSigma$ and the top $r$ eigenvectors $\hat \u_1,\dots,\hat \u_r$ of $\what \bSigma$ converge almost surely to the following limits:
\begin{align}
c_{jm} = 
\begin{cases}
\sqrt{\frac{\ell_j^2 - \gamma}{\ell^2 + \gamma \ell_j}}
    & \text{ if } j=m \text{ and } \ell_j > \sqrt{\gamma}  \\
0, & \text{ otherwise }
\end{cases}.
\end{align}
Proofs of these relationships may be found in the seminal work \cite{paul2007asymptotics}, while extensions to more general noise models can be found in \cite{benaych2012singular}.  These relationships have been used to devise estimators of $\bSigma$, as in the work \cite{donoho2018optimal}, as well as to predict the signal vectors $X_1,\dots,X_l$ themselves, as in \cite{shabalin2013reconstruction}, \cite{gavish2017optimal}, \cite{leeb2020operator}, \cite{leeb2019matrix}.

Of particular significance for the present work is the value $(1 + \sqrt{\gamma})^2$. This is the limiting value of the operator norm of the sample covariance $\frac{1}{l}\sum_{j=1}^{l} \ep_j \ep_j^T$ of the noise component alone as $k,l \to \infty$ and $k/l \to \gamma$; see \cite{geman1980limit}. More generally, for any sequence of random noise matrices whose operator norms have an almost sure asymptotic limit as the dimensions grow to infinity, we refer to that limiting value as the \emph{spectral signal detection threshold}, or \emph{SSDT} for short. As is apparent for white noise from the formula \eqref{eq:spike}, observed eigenvalues exceeding the SSDT may be attributed to the signal component, and used to extract information about the signal. Eigenvalues below the SSDT are not as useful for recovering information on the signal (though see Remark \ref{rmk:detection} below). In typical applications of the spiked covariance model and its generalizations to signal denoising, only the eigenvalues of $\what \bSigma$ exceeding the SSDT are used for denoising \cite{gavish2017optimal}, \cite{donoho2018optimal}, \cite{leeb2019matrix}, \cite{nadakuditi2014optshrink}, \cite{dobriban2020optimal}, \cite{leeb2019optimal}, \cite{bhamre2016denoising}. For this reason, evaluation of the SSDT is a critical task in statistical signal processing.


While the SSDT and related quantities, such as the relations between the eigenvalues of $\bSigma$ and $\what \bSigma$, have simple formulas in the case of white noise, more general noise models pose greater challenges. In this paper, we study a fairly broad family of noise matrices, namely those with a \emph{separable variance profile}: every entry of the noise matrix has a potentially different variance, but the matrix of all the variances is rank one.
This type of random matrix arises naturally in a number of applications, such as signal denoising with variable-strength heteroscedastic noise and the Kronecker model of multiple-input multiple-output transmission. We show that even when there are not closed formulas for the quantities of interest in this model, they may nevertheless be evaluated rapidly and to machine precision via provably correct algorithms.

\subsection{This paper's contributions}

This paper introduces fast, scalable, and numerically stable algorithms for two related problems arising in statistical signal processing applications. The first is to evaluate the spectral signal detection threshold (SSDT) for noise matrices with a separable variance profile.
As explained above, the SSDT is the asymptotic value of the operator norm of a random matrix representing the noise component in a signal-plus-noise observation model as the dimensions of the matrix grow to infinity.

The second problem we address is evaluating the Stieltjes transform (also known as the Cauchy-Stieltjes transform) of a certain probability measure associated with the noise matrix, known as the limiting spectral distribution (LSD). The Stieltjes transform is a function of central interest in random matrix theory and its applications; we recall its precise definition in Section \ref{sec-setting}. As we will review in Section \ref{sec-dtransform}, the Stieltjes transform, evaluated at real values exceeding the SSDT, may be used to estimate certain model parameters and perform signal denoising. Like the SSDT itself, in isotropic noise the Stieltjes transform has a closed form, while more general variance profiles pose greater challenges which we address in the current paper.


The algorithms we present scale linearly with the number of problem parameters, and are provably accurate to machine precision. The solution and analysis of each problem makes use of the master equations characterizing the Stieltjes transform, which will be reviewed in Section \ref{sec-stieltjes}. Through a detailed analysis, presented in Section \ref{sec-math}, we prove that the desired values may be computed by the use of Newton's root-finding algorithm. For the SSDT, several intermediate quantities used in each iteration of Newton's method are themselves computable by Newton's method, as we will show. Consequently, all parameters we compute are either available analytically, or can be provably computed to full precision by a rapidly converging iterative scheme. For matrices of size $p$-by-$n$, the asymptotic costs of the algorithms scale like $O(p+n)$.

\subsection{Relation to prior work}

The signal-plus-noise matrix models of the kind we consider have been previously studied in the statistical literature, in works such as \cite{johnstone2001distribution},  \cite{benaych2012singular},  \cite{nadakuditi2014optshrink},  \cite{dobriban2020optimal},  \cite{leeb2019optimal},  \cite{gavish2017optimal},  \cite{paul2007asymptotics},  \cite{leeb2019matrix},  \cite{donoho2018optimal},  \cite{hong2018optimally},  \cite{hong2016towards},  \cite{zhang2018heteroskedastic},  \cite{hong2018asymptotic}.  Estimating the number of signal terms in such principal components analysis and factor models is a well-studied problem in statistics and statistical signal processing applications \cite{dobriban2017deterministic},  \cite{dobriban2017factor},  \cite{buja1992remarks},  \cite{horn1965rationale},  \cite{kritchman2008determining},  \cite{passemier2012determining},  \cite{kritchman2009non}. Detection in the low SNR regime constitutes a distinct though conceptually relevant class of problems in signal processing \cite{tandra2008snr},  \cite{yucek2009survey},  \cite{zeng2007covariance},  \cite{zeng2009eigenvalue}.  Noise matrices with separable variance profiles have also been studied in the context of the Kronecker product model for MIMO wireless communication \cite{couillet2011random},  \cite{biglieri2007mimo} and factor analysis with applications to economics \cite{onatski2010determining}. We also note that such random matrices are a special case of ``algebraic'' random matrices, introduced in \cite{rao2008polynomial}.

The solution to both of the problems we study rests on a known characterization of the Stieltjes transform of the LSD as the solution to a set of certain non-linear equations; these have appeared in \cite{paul2009no} and  \cite{couillet2014analysis}. We present a new, detailed analysis of these equations. As a consequence of our analysis, we show that evaluating the Stieltjes transform may be done by a straightforward application of Newton's root-finding algorithm, whereas the SSDT is computable by several nested applications of Newton's agorithm, appropriately initialized.

Previous works have considered the problem of evaluating the LSD, its Stieltjes transform, and the boundary of its support, for different classes of random noise matrices. The paper \cite{dobriban2015efficient} proposes a scheme for evaluating the Stieltjes transform at complex values outside the support of the LSD; this method is based on a non-linear equation satisfied by the Stieltjes transform \cite{silverstein1995strong},  \cite{silverstein1995empirical},  \cite{marchenko1967distribution}. The method in \cite{cordero2018mixandmix} is devoted to extending the approach to mixture models. The papers \cite{ledoit2015spectrum},  \cite{ledoit2017numerical} are concerned with solving the inverse problem of recovering the population distribution from the observed spectrum; this problem is also taken up in \cite{elkaroui2008spectrum}. The paper \cite{dobriban2017deterministic} contains a method for finding the boundary of a certain family of LSDs based on root-finding, although the use of Newton's method is not employed or analyzed. To the best of our knowledge, the problem of evaluating the SSDT and Stieltjes transform for random matrices with a separable variance profile has not been addressed in prior work; and for special cases that have been studied, such as in \cite{dobriban2017deterministic}, the prior work does not contain fast algorithms with the guarantees presented here.


\subsection{Outline}

The remainder of the paper is outlined as follows. In Section \ref{sec-prelim}, we precisely state the problems we solve in this paper, and review the mathematical and numerical material that we will be using. In Section \ref{sec-math}, we derive the core mathematical theory on which our algorithms rest, namely a detailed analysis of the master equations characterizing the Stieltjes transform of the LSD. In Section \ref{sec-algorithms}, we provide explicit descriptions of the numerical algorithms for finding the SSDT and evaluating the Stieltjes transform. In Section \ref{sec-numerical}, we present the results of several illustrative numerical experiments demonstrating the performance of our algorithms.

\section{Preliminaries}
\label{sec-prelim}

\subsection{Setting and problem formulation}
\label{sec-setting}

We suppose we have a $k$-by-$l$ random matrix of the form $\N = \A^{1/2} \G \B^{1/2}$, where $\G$ is a random matrix with iid entries of mean zero and variance $l^{-1}$, and $\A$ and $\B$ are positive-definite matrices of sizes $k$-by-$k$ and $l$-by-$l$, respectively. We define the \emph{empirical spectral distribution (ESD)} $\mu_k$ to be the distribution of eigenvalues $\lambda_1,\dots,\lambda_k$ of the matrix $\N \N^T$:
\begin{align}
d \mu_k(t) = \frac{1}{k} \sum_{i=1}^{k} \delta_{\lambda_i}(t).
\end{align}

As is typical in high-dimensional problems, we work in the asymptotic regime where $k$ and $l$ both grow to infinity. More precisely, we suppose that $l = l(k)$ grows with $k$, and that the limit
\begin{align}
\gamma = \lim_{k \to \infty} \frac{k}{l(k)}
\end{align}
is well-defined, positive, and finite. Furthermore, as $k$ and $l$ grow, we suppose that the eigenvalue spectra of $\A$ and $\B$ have well-defined asymptotic distributions $\nu$ and $\unu$, respectively (in terms of weak convergence), with compact support. Under suitable conditions on $\A$ and $\B$, the sequence of measures $\mu_k$ will almost surely converge weakly to a compactly supported measure $\mu$, called the \emph{limiting spectral distribution (LSD)}. We also denote by  $\umu$ the limiting spectral distribution of the eigenvalues of $\N^T \N$.

We define the \emph{spectral signal detection threshold (SSDT)} as follows:
\begin{align}
\lambda^* = \operatorname*{arg\,max}\{ \lambda > 0 : \mu([\lambda,\infty)) > 0 \}.
\end{align}
This is the right endpoint of the LSD $\mu$'s support, or equivalently the asymptotic operator norm of the matrix $\N \N^T$. As we will explain in Section \ref{sec-dtransform}, in a signal-plus-noise model $\Y = \X + \N$ where $\X$ is a low-rank signal matrix, eigenvalues exceeding $\lambda^*$ are attributable to the signal $\X$.


Next, we define the \emph{Stieltjes transform} of $\mu$:
\begin{align}
\label{eq-s}
s(\lambda) = \int_{\R} \frac{1}{t - \lambda} d\mu(t),
\end{align}
which has derivative equal to
\begin{align}
\label{eq-sder}
s'(\lambda) = \int_{\R} \frac{1}{(t - \lambda)^2} d\mu(t).
\end{align}
Similarly, the Stieltjes transform of $\umu$ is given by
\begin{align}
\label{eq-sbar}
\s(\lambda) = \int_{\R} \frac{1}{t -\lambda} d\umu(t),
\end{align}
with derivative equal to
\begin{align}
\label{eq-sbarder}
\s'(\lambda) = \int_{\R} \frac{1}{(t - \lambda)^2} d\umu(t).
\end{align}
We call $\s(\lambda)$ the \emph{associated Stieltjes transform} of $\mu$. The functions $s$ and $\s$ are defined for all complex $\lambda$ outside the supports of $\mu$ and $\umu$, respectively. However, as we will explain in Section \ref{sec-dtransform}, in the present work we are only interested in evaluating $s(\lambda)$, $\s(\lambda)$, and their derivatives for real $\lambda > \lambda^*$.

In this paper, we consider the setting where $\nu$ and $\unu$ are discrete distributions of the form
\begin{align}
d \nu =\sum_{i=1}^{p} \omega_i \delta_{a_i}
\end{align}
and
\begin{align}
d \unu = \sum_{j=1}^{n} \pi_j \delta_{b_j}.
\end{align}
Here, $a_1,\dots,a_p$ and $b_1,\dots,b_n$ are positive numbers, and the weights $\omega_i$ and $\pi_j$ satisfy $\sum_{i=1}^{p} \omega_i = \sum_{j=1}^{n} \pi_j = 1$, $\omega_i > 0$, and $\pi_j > 0$. For example, if $k/2$ eigenvalues of $\A$ are 1, and the remaining $k/2$ are equal to 2, then $p=2$, $a_1 = 1$, $a_2=2$, and $\omega_1 = \omega_2 = 1/2$.

With this background and notation, we can concisely state the two problems we solve in this paper:


\vspace{\baselineskip}

\noindent
\textbf{Problem 1.}
Given $a_1,\dots,a_p$ with corresponding weights $\omega_1,\dots,\omega_p$; and $b_1,\dots,b_n$ with corresponding weights $\pi_1,\dots,\pi_n$; evaluate the SSDT $\lambda^*$.

\vspace{\baselineskip}

\noindent
\textbf{Problem 2.}
Given $a_1,\dots,a_p$ with corresponding weights $\omega_1,\dots,\omega_p$; and $b_1,\dots,b_n$ with corresponding weights $\pi_1,\dots,\pi_n$; and a value $\lambda > \lambda^*$; evaluate $s(\lambda)$, $s'(\lambda)$, $\s(\lambda)$, and $\s'(\lambda)$.

\vspace{\baselineskip}

Algorithms solving Problems 1 and 2 to machine precision, with asymptotic cost $O(p + n)$, are presented in Section \ref{sec-algorithms}.

\begin{rmk}
In the problem statements, the parameters $p$, $n$, $a_1,\dots,a_p$, and $b_1,\dots,b_n$ are user-supplied inputs. In many statistical applications, these parameters are assumed to be known, as in the works \cite{dobriban2017deterministic}, \cite{leeb2019optimal}, \cite{dobriban2015efficient}, \cite{hong2018optimally}, \cite{hong2016towards},  \cite{hong2018asymptotic}.
\end{rmk}

\begin{rmk}
The assumption that $d\nu$ and $d\unu$ are discrete measures is quite common in the literature on random matrix theory and its applications \cite{dobriban2017deterministic}, \cite{leeb2019optimal}, \cite{dobriban2015efficient}, \cite{hong2018optimally},  \cite{hong2016towards},  \cite{hong2018asymptotic}. It is likely that the methods in the present work can be generalized to more general measures, so long as certain linear functionals of the measures can be efficiently computed (e.g.\ by numerical quadrature). Pursuing this in detail lies outside the scope of the current work, however.
\end{rmk}


%

%

\subsection{Properties of the Stieltjes transform}
\label{sec-stieltjes}
The Stieltjes and associated Stieltjes transforms of the LSD $\mu$ are defined by equations \eqref{eq-s} and \eqref{eq-sbar}, respectively. We refer the reader to the standard references \cite{bai2009spectral},  \cite{tao2012topics} for a detailed treatment of Stieltjes transforms of probability measures, particularly random spectral measures.

\begin{lem}
\label{lem-s2sbar}
The Stieltjes transforms $s(\lambda)$ and $\s(\lambda)$ satisfy the following relations:
\begin{align}
\label{eq-s2sbar}
\s(\lambda) = \gamma s(\lambda) + \frac{\gamma-1}{\lambda}
\end{align}
and
\begin{align}
\label{eq-sbar2s}
s(\lambda) = \frac{1}{\gamma}  \s(\lambda)
    + \bigg(\frac{1}{\gamma} - 1\bigg) \frac{1}{\lambda}.
\end{align}
\end{lem}

\begin{lem}
\label{lem-s2sbar-der}
The derivatives $s'(\lambda)$ and $\s'(\lambda)$ satisfy the following relations:
\begin{align}
\label{eq-s2sbar-der}
\s'(\lambda) = \gamma s'(\lambda) + \frac{1-\gamma}{\lambda^2}
\end{align}
and
\begin{align}
\label{eq-sbar2s-der}
s'(\lambda) = \frac{1}{\gamma}  \s'(\lambda)
    + \bigg(1 - \frac{1}{\gamma}\bigg) \frac{1}{\lambda^2}.
\end{align}
\end{lem}

Lemma \ref{lem-s2sbar-der} follows from Lemma \ref{lem-s2sbar}, which in turn follows immediately from the relation:
\begin{align}
d \umu(t) = \gamma d\mu(t) + (1-\gamma) \delta_0(t).
\end{align}

\begin{rmk}
Lemmas \ref{lem-s2sbar} and \ref{lem-s2sbar-der} show that $\s(\lambda)$ and $\s'(\lambda)$ may be easily evaluated once $s(\lambda)$ and $s'(\lambda)$ are computed.
\end{rmk}

%



The next result is central to our subsequent analysis.

\begin{thm}
\label{thm-master}
The Stieltjes transform $s(\lambda)$ for the LSD satisfies the following \emph{master equations}:
\begin{align}
s(\lambda) = \int_{\R} \frac{1}{a G(e(\lambda)) - \lambda} d\nu (a),
\end{align}
where 
\begin{align}
G(e) = \int_{\R} \frac{b}{1 + \gamma b e} d \unu(b)
\end{align}
and $e(\lambda)$ is a function that satisfies the equation
\begin{align}
e(\lambda) = \int_{\R} \frac{a}{a G(e(\lambda)) - \lambda} d\nu (a).
\end{align}
\end{thm}

For a proof of Theorem \ref{thm-master}, see, for instance, the paper \cite{paul2009no}. The paper \cite{couillet2014analysis} presents a slightly modified form of these equations, with a detailed analysis showing that $e(\lambda)$ is smooth for real $\lambda$ outside the support of $\mu$.

We assume that $\nu$ and $\unu$ are discrete measures of the form
\begin{align}
d \nu =\sum_{i=1}^{p} \omega_i \delta_{a_i},
\end{align}
and
\begin{align}
d \unu = \sum_{j=1}^{n} \pi_j \delta_{b_j},
\end{align}
where $\sum_{i=1}^{p} \omega_i = \sum_{j=1}^{n} \pi_j = 1$, $\omega_i > 0$, and $\pi_j > 0$.
The master equations therefore become:
\begin{align}
\label{eq-e2s}
s(\lambda) = \sum_{i=1}^{p} \frac{\omega_i}{a_i G(e(\lambda)) - \lambda}
\end{align}
where $e(\lambda)$ is a function satisfying
\begin{align}
\label{eq-implicit-e}
e(\lambda) = \sum_{i=1}^{p} \frac{a_i \omega_i }{a_i G(e(\lambda)) - \lambda}
\end{align}
and the function $G$ is defined by:
\begin{align}
\label{eq-gdef}
G(e) = \sum_{j=1}^{n} \frac{b_j \pi_j}{1 + \gamma b_j e}.
\end{align}

We note that in the case where $\B= \I_n$ (the singly-weighted case), the master equations take on a simpler form; see, for instance, \cite{silverstein1995empirical},  \cite{silverstein1995strong}.

\subsection{Spiked models and the $D$-transform}
\label{sec-dtransform}

In a spiked matrix model, we observe a signal-plus-noise matrix $\Y$ of the form
\begin{align}
\Y = \X + \N,
\end{align}
where $\X = \sum_{m=1}^{r} \theta_m \u_m \v_m^T$ is a rank $r \ll \min\{k,l\}$ signal matrix, and $\N$ is a noise matrix. The $D$-transform is defined as follows:
\begin{align}
D(\lambda) =  \lambda s(\lambda) \s(\lambda),
\end{align}
where $s(\lambda)$ and $\s(\lambda)$ are, respectively, the Stieltjes transform and associated Stieltjes transform of the LSD of $\N \N^T$.  The $D$-transform was introduced in \cite{benaych2012singular}, though as a function of $\sqrt{\lambda}$ rather than $\lambda$.

Denoting by $\lambda_1,\dots,\lambda_r$ the top $r$ eigenvalues of $\Y \Y^T$, and $\hat \u_1,\dots, \hat \u_r$, $\hat \v_1,\dots, \hat \v_r$ the corresponding left and right singular vectors of $\Y$, respectively, it is shown in \cite{benaych2012singular} that the $D$-transform defines a mapping between $\lambda_1,\dots,\lambda_r$ and the eigenvalues $\theta_1^2,\dots,\theta_r^2$ of the signal matrix $\X \X^T$, which holds almost surely in the limit $k,l \to \infty$:
\begin{align}
\theta_m^2 = \lim_{k,l \to \infty}\frac{1}{D(\lambda_m)}.
\end{align}
This is satisfied for sufficiently large eigenvalues $\theta_m^2$ of $\X \X^T$, namely those for which $\theta_m^2 > 1/D(\lambda^*)$. Phrased differently, an observed eigenvalue $\lambda_m$ of $\Y \Y^T$ satisfying $\lambda_m > \lambda^*$ may be attributed to the presence of signal.


Furthermore, the asymptotic cosines $\la \u_m, \hat \u_m \ra$ and $\la \v_m, \hat \v_m \ra$, $1 \le m \le r$, can also be evaluated using the Stieltjes transform. It is shown that, almost surely,
\begin{align}
\lim_{k,l \to \infty} |\la \u_m, \hat \u_m \ra|^2 
    = \frac{s(\lambda_m) D(\lambda_m)}{D'(\lambda_m)}
\end{align}
and
\begin{align}
\lim_{k,l \to \infty} |\la \v_m, \hat \v_m \ra|^2 
    = \frac{\s(\lambda_m) D(\lambda_m)}{D'(\lambda_m)}
\end{align}
Consequently, evaluating $s(\lambda_m)$ and $s'(\lambda_m)$ provides a method for estimating the angles between the population singular vectors $\u_m$, $\v_m$ and the observed singular vectors  $\hat \u_m$, $\hat \v_m$. These relationships have been employed to derive optimal methods of singular value shrinkage \cite{nadakuditi2014optshrink}.

\begin{rmk}
\label{rmk:detection}
While eigenvalues of $\Y \Y^T$ exceeding the SSDT $\lambda^*$ indicate the presence of signal, the converse statement -- if all eigenvalues are less than $\lambda^*$, then there is no signal -- is subtler. For example, the work \cite{onatski2013asymptotic} studies the detection of signals in white noise and shows that the joint distribution of all eigenvalues of $\Y \Y^T$ changes even when the signal eigenvalues do not exceed the SSDT. However, the resulting statistical test does not accurately distinguish between the presence or absence of signal with overwhelming probability. Regardless, in tasks such as covariance and signal estimation in the spiked model, only those eigenvalues exceeding the SSDT are typically used for covariance estimation and signal denoising \cite{gavish2017optimal}, \cite{donoho2018optimal}, \cite{leeb2019matrix}, \cite{nadakuditi2014optshrink}, \cite{dobriban2020optimal}, \cite{leeb2019optimal}, \cite{bhamre2016denoising}.
\end{rmk}

\begin{rmk}
Separable variance profiles arise naturally when the columns of $\Y$ are independent random vectors in $\R^p$ of the form
\begin{align}
Y_j = X_j + b_j^{1/2} \A^{1/2} G_j,
\quad 1 \le j \le n,
\end{align}
where the $X_j$ are signal vectors constrained to an $r$-dimensional subspace of $\R^p$; $\A$ is a positive definite matrix; and $b_1,\dots,b_n$ are specified weights. In this model, each signal vector $X_j$ is observed in the presence of heteroscedastic noise (with covariance $\A$) of variable strength $b_j$. The noise matrix alone is distributed like $\N = \A^{1/2} \G \B^{1/2}$, where $\B = \diag(b_1,\dots,b_n)$. Models like these are considered in, for example, \cite{leeb2019matrix}, \cite{ding2021spiked}. When $\A = \I_p$, this model is considered in \cite{hong2018optimally},  \cite{hong2016towards},  \cite{hong2018asymptotic}; when $\B = \I_n$, this model is considered in \cite{zhang2018heteroskedastic}, \cite{leeb2019optimal}.
\end{rmk}

\begin{prop}
\label{prop-dtransform}
For $\lambda > \lambda^*$, $D(\lambda)$ is decreasing and convex.
\end{prop}
\begin{proof}
Write $D(\lambda) = \varphi(\lambda) \underline{\varphi}(\lambda)$, where $\varphi(\lambda) = \sqrt{\lambda} s(\lambda)$ and $\underline{\varphi}(\lambda) = \sqrt{\lambda} \s(\lambda)$. Using $D'(\lambda) = \varphi(\lambda) \underline{\varphi}'(\lambda) + \varphi'(\lambda) \underline{\varphi}(\lambda)$ and $D''(\lambda) = \varphi(\lambda) \underline{\varphi}''(\lambda) + \varphi''(\lambda) \underline{\varphi}(\lambda) + 2 \varphi'(\lambda) \underline{\varphi}'(\lambda)$, and $\varphi(\lambda) < 0$ and $\underline{\varphi}(\lambda) < 0$ for $\lambda > \lambda^*$, it is enough to show that $\varphi(\lambda)$ $\underline{\varphi}(\lambda)$ are increasing and concave. But this follows immediately from the identities
\begin{align}
\frac{d}{d\lambda} \frac{\sqrt{\lambda}}{t - \lambda}
= \frac{1}{2} \frac{\lambda + t}{\sqrt{\lambda} (t - \lambda)^2} > 0
\end{align}
and
\begin{align}
\frac{d^2}{d\lambda^2} \frac{\sqrt{\lambda}}{t - \lambda}
= \frac{1}{4} \frac{6\lambda t + 3\lambda^2 - t^2}{\lambda^{3/2} (t - \lambda)^3} < 0,
\end{align}
and the definitions \eqref{eq-s} and \eqref{eq-sbar} of $s(\lambda)$ and $\s(\lambda)$.
\end{proof}

\subsection{Newton's root-finding algorithm}
\label{sec-newton}

Newton's method is a classical technique for finding the roots of a function of one real variable. We briefly review the method here; the reader may consult any standard reference on optimization or numerical analysis, such as \cite{nesterov2018lectures},  \cite{dahlquist1874numerical}, for additional details. We are given a smooth function $f(x)$, where $x \in [a,b]$, and we suppose $f(a) < 0$ and $f(b) > 0$. We also suppose that $f'(x) > 0$ and $f''(x) < 0$ for all $x \in (a,b)$; that is, $f$ is a strictly increasing and concave. Our goal is to compute $x^*$, the unique root of $f$ in $(a,b)$.

To find $x^*$, Newton's method initializes $a < x_0 < x^*$, and defines a sequence of updates recursively as follows: given an estimate $x_k$, the next value $x_{k+1}$ is defined by
\begin{align}
x_{k+1} = x_k - \frac{f(x_k)}{f'(x_k)}.
\end{align}
Geometrically, $x_{k+1}$ is the root of the line tangent to the graph of $f$ at the point $(x_k,f(x_k))$. Because $f$ is concave, it is easy to see that $x_k < x_{k+1} \le x^*$. 


Because each $x_k$ is obtained from a linear approximation to $f$, the errors decay quadratically in the vicinity of $x^*$; that is
\begin{align}
|x_{k+1} - x^*| \le C |x_k - x^*|^2
\end{align}
when $|x_k - x^*|$ is sufficiently small.

\begin{rmk}
Quadratic convergence is what makes Newton's method an especially attractive algorithm when it is applicable -- that is, when both $f(x)$ and $f'(x)$ may be accurately evaluated, and $f$ exhibits the conditions specified above. In practical terms, quadratic convergence means that the number of accurately computed digits of $x^*$ approximately doubles after each iteration of the algorithm, until machine precision is reached.
\end{rmk}

\section{Mathematical apparatus}
\label{sec-math}

In this section, we analyze the master equations \eqref{eq-e2s} -- \eqref{eq-gdef}. Our results will provide the necessary tools to devise algorithms for computing the SSDT $\lambda^*$ and evaluating $s(\lambda)$ and $s'(\lambda)$ for $\lambda > \lambda^*$. 
We define 
\begin{align}
a^* = \max_{1 \le i \le p} a_i,
\end{align}
and
\begin{align}
b^* = \max_{1 \le j \le n} b_j.
\end{align}

We define the function $F(\lambda,e)$ by:
\begin{align}
F(\lambda,e) = e - \sum_{i=1}^{p} \frac{a_i \omega_i}{a_i G(e) - \lambda}.
\end{align}
Then for each $\lambda > \lambda^*$, $e(\lambda)$ satisfies $F(\lambda,e(\lambda)) = 0$. When we treat $\lambda$ as a fixed parameter and $e$ as a variable, we will use the notation $F_\lambda(e) = F(\lambda,e)$.


%

\subsection{Range and monotonicity of $e(\lambda)$ when $\lambda > \lambda^*$}
We first state a result on the range of $e(\lambda)$ for $\lambda > \lambda^*$. The function $G(e)$ approaches $0$ as $e \to \infty$, and grows to $+\infty$ as $e \to (-1/\gamma b^*)^+$, and is strictly decreasing on the interval 
\begin{align}
J \equiv \left\{ e : e > \frac{-1}{\gamma b^*} \right\}.
\end{align}
For any $\lambda > 0$, we define the interval $I_\lambda$ by
\begin{align}
I_\lambda \equiv \left\{ e \in J : G(e) < \frac{\lambda}{a^*} \right\}.
\end{align}

\begin{prop}
\label{prop-range-e}
When $\lambda > \lambda^*$, $e(\lambda)$ is contained in the interval $I_\lambda \cap (-\infty,0)$.
\end{prop}

We will develop the proof in several steps.

\begin{lem}
Let $\epsilon > 0$. Then the function $G(e(\lambda))$ is bounded for all $\lambda > \lambda^* + \epsilon$.
\end{lem}
\begin{proof}
We show that the range of $e(\lambda)$ cannot approach any of the singularities of $G$, which lie at the values $-1/(\gamma b_j)$. Define 
\begin{align}
H(\lambda) = \frac{G(e(\lambda))}{\lambda}.
\end{align}
Then from \eqref{eq-e2s}
\begin{align}
\lambda s(\lambda) = \sum_{i=1}^{p} \frac{\omega_i}{a_i H(\lambda) - 1},
\end{align}
and consequently
\begin{align}
\label{eq123}
1 + \lambda s(\lambda) = \sum_{i=1}^{p} \frac{\omega_i}{a_i H(\lambda) - 1}
    + \sum_{i=1}^{p} \omega_i\frac{a_i H(\lambda) - 1}{a_i H(\lambda) - 1}
= \sum_{i=1}^{p} \frac{a_i H(\lambda) \omega_i }{a_i H(\lambda) - 1}
= G(e(\lambda)) e(\lambda).
\end{align}
Since $\lambda s(\lambda)$ is bounded for $\lambda > \lambda^* + \epsilon$, this tells us that $G(e(\lambda))$ must stay bounded so long as $e(\lambda)$ is bounded away from $0$; in particular, $e(\lambda)$ cannot be made arbitrarily close to any of the singularities $-1/(\gamma b_j)$.
\end{proof}

\begin{cor}
$e(\lambda) \to 0^-$ as $\lambda \to \infty$.
\end{cor}
\begin{proof}
Because $G(e(\lambda))$ is bounded for large $\lambda$, the result follows from \eqref{eq-implicit-e}.
\end{proof}

\begin{cor}
\label{cor-123}
For any $\lambda > \lambda^*$, $e(\lambda) \in J$; that is,
\begin{align}
e(\lambda) > \frac{-1}{\gamma b^*}.
\end{align}
\end{cor}
\begin{proof}
This follows from the continuity of $e(\lambda)$ for $\lambda > \lambda^*$, and the facts that it never passes through $-1/(\gamma b_j)$ and approaches 0 at large $\lambda$.
\end{proof}

\begin{cor}
For all $\lambda > \lambda^*$, $G(e(\lambda)) > 0$ and $e(\lambda) < 0$.
\end{cor}
\begin{proof}
The positivity of $G(e(\lambda))$ follows immediately from Corollary \ref{cor-123}. The negativity of $e(\lambda)$ then follows from \eqref{eq123}, and the fact that $1 + \lambda s(\lambda) < 0$.
\end{proof}

\begin{lem}
For all $\lambda > \lambda^*$,
\begin{align}
\frac{G(e(\lambda))}{\lambda} <  \frac{1}{a^*}.
\end{align}
\end{lem}
\begin{proof}
We have $G(e(\lambda)) \ne \lambda / a_i$, from \eqref{eq-e2s}. Suppose for some $\lambda > \lambda^*$, we had
\begin{align}
\frac{G(e(\lambda))}{\lambda} > \frac{1}{a_i}.
\end{align}
The inequality must then remain true for all sufficiently large $\lambda$, since $G(e(\lambda)) / \lambda$ is continuous and does not pass through $1/a_i$. However, the left side converges to $0$ as $\lambda \to \infty$, since $G(e(\lambda))$ is bounded for large $\lambda$; a contradiction. This completes the proof.
\end{proof}

We have shown that for all $\lambda > \lambda^*$, $e(\lambda)$ lies in the interval defined by the inequalities
\begin{align}
G(e) \le  \frac{\lambda}{a^*},
    \quad \quad
e > \frac{-1}{\gamma b^*}, 
    \quad \quad
e < 0.
\end{align}

This completes the proof of Proposition \ref{prop-range-e}.

Next we prove that $e(\lambda)$ is increasing:

\begin{prop}
The function $e(\lambda)$ is increasing for $\lambda > \lambda^*$.
\end{prop}
\begin{proof}
We have:
\begin{align}
(\partial_\lambda F) (\lambda,e)
= -\sum_{i=1}^{p} \frac{a_i \omega_i}{(a_i G(e) - \lambda)^2} < 0.
\end{align}
Since $F(\lambda,e(\lambda)) = 0$, we have:
\begin{align}
\label{eq12345}
0 = \frac{\partial}{\partial \lambda} \{ F(\lambda,e(\lambda))\}
= (\partial_\lambda F)(\lambda, e(\lambda))
    + e'(\lambda) (\partial_e F ) (\lambda, e(\lambda)) ,
\end{align}
and so
\begin{align}
\label{eq4567}
e'(\lambda) (\partial_e F)  (\lambda, e(\lambda))
= \sum_{i=1}^{p} \frac{a_i \omega_i}{(a_i G(e) - \lambda)^2} > 0.
\end{align}
Consequently, $e'(\lambda)$ can never be $0$. Furthermore,
\begin{align}
(\partial_e F) (\lambda, e)
= 1 + G'(e(\lambda)) 
    \sum_{i=1}^{p} \left( \frac{a_i}{a_i G(e(\lambda)) - \lambda} \right)^2 \omega_i
\end{align}
which converges to 1 as $\lambda \to \infty$ (note that $e(\lambda)$ stays bounded away from singularities of $G(e)$ and also $G'(e)$, which have the same singularities). So $e'(\lambda) > 0$ for sufficiently large $\lambda$, and hence, since it is continuous and cannot pass through 0, $e'(\lambda) > 0$ for all $\lambda > \lambda^*$.
\end{proof}

\subsection{Behavior of $F_\lambda(e)$}
\label{sec-fbehavior}

In this section we characterize the behavior of $F_\lambda(e) = F(\lambda,e)$ (viewed as a function of $e$) on the interval $I_\lambda$. Specifically, we show the following. For any $\lambda > 0$, $F_\lambda(e)$ is a strictly convex function that approaches $+\infty$ as $e$ approaches either end of $I_\lambda$. Furthermore, when $\lambda > \lambda^*$, the minimum value of $F_\lambda(e)$ is less than zero, and $F_\lambda(0) > 0$; consequently, there are exactly two roots of $F_\lambda(e)$, both contained in the interval $I_\lambda \cap (-\infty,0)$. We show that $e(\lambda)$ is always equal to the largest root, namely the one at which $F_\lambda'(e) > 0$.

\begin{lem}
For $\lambda > 0$, the function $F_\lambda(e)$ is strictly convex on $I_\lambda$; that is, $(\partial_{ee}^2 F)(\lambda,e) > 0$.
\end{lem}
\begin{proof}
We have:
\begin{align}
(\partial_e F) (\lambda, e)
= 1 + G'(e(\lambda)) 
    \sum_{i=1}^{p} \left( \frac{a_i}{a_i G(e(\lambda)) - \lambda} \right)^2 \omega_i
\end{align}
and
\begin{align}
(\partial_{ee}^2 F) (\lambda, e)
= G''(e)\sum_{i=1}^{p}\left( \frac{a_i}{a_i G(e) - \lambda} \right)^2 \omega_i
    - 2G'(e)^2\sum_{i=1}^{p}\left( \frac{a_i}{a_i G(e) - \lambda} \right)^3 \omega_i.
\end{align}
Now whenever $e \in I_\lambda$, we have $e > -1/\gamma b^* \ge -1/\gamma b_j$ for all $1 \le j \le n$, and $G(e) < \lambda/a^* \le \lambda / a_i$ for all $1 \le i \le p$; consequently,
\begin{align}
G''(e) = 2\gamma^2 \sum_{j=1}^{n} \left( \frac{b_j}{1+\gamma b_j e} \right)^3 \pi_j> 0
\end{align}
and
\begin{align}
\sum_{i=1}^{p}\left( \frac{a_i}{a_i G(e) - \lambda} \right)^3 \omega_i < 0.
\end{align}
Consequently, $(\partial_{ee}^2 F) (\lambda, e) > 0$ for all $e \in I_\lambda$, i.e.\ the function $F_\lambda(e)$ is convex.
\end{proof}

\begin{lem}
The function $F_\lambda(e)$ diverges to $+\infty$ as $e \to +\infty$ and as $e$ approaches the left endpoint of $I_\lambda$ from the right.
\end{lem}
\begin{proof}
This is immediate from the definition of $F(\lambda,e)$.
\end{proof}

\begin{lem}
For each $\lambda > \lambda^*$,  $F(\lambda,0) > 0$.
\end{lem}
\begin{proof}
We have:
\begin{align}
F(\lambda,0) =  - \sum_{i=1}^{p} \frac{a_i \omega_i}{a_i G(0) - \lambda}.
\end{align}
Since $G(e(\lambda)) < \lambda / a_i$ and $G(e)$ is decreasing on $I_\lambda$, and $e(\lambda) < 0$, we also have $G(0) < G(e(\lambda)) < \lambda / a_i$; consequently, $F(\lambda,0) > 0$.
\end{proof}

\begin{prop}
\label{prop-identify}
For $\lambda > \lambda^*$, $(\partial_e F)(\lambda,e(\lambda)) > 0$.
\end{prop}
\begin{proof}
This follows from \eqref{eq4567} and $e'(\lambda) > 0$.
\end{proof}

We have shown that $F_\lambda(e)$ has two roots in the interval $I_\lambda \cap (-\infty,0)$ whenever $\lambda > \lambda^*$. Proposition \ref{prop-identify} identifies $e(\lambda)$ as the root that is closest to zero, or equivalently, the root where the derivative of $F_\lambda$ is positive. This characterization will be used in Section \ref{sec-evaluate} to devise the algorithm for computing $e(\lambda)$, and consequently $s(\lambda)$.

\subsection{The minimum of $F_\lambda(e)$}



We will let $t(\lambda)$ denote the minimum of $F_\lambda(e)$ on $J_\lambda$; that is, $t(\lambda)$ is the unique value on $J_\lambda$ that satisfies
\begin{align}
(\partial_e F)(\lambda, t(\lambda)) = 0.
\end{align}

We define the function $Q(\lambda)$ for $\lambda > 0$ by:
\begin{align}
Q(\lambda) = F(\lambda,t(\lambda)).
\end{align}
For any $\lambda > 0$, we define the function $R_\lambda(e)$ for $e \in I_\lambda$ by:
\begin{align}
R_\lambda(e) = (\partial_e F)(\lambda,e) = F_\lambda'(e).
\end{align}
Note that by definition, $R_\lambda(t(\lambda)) = 0$ for all $\lambda = 0$.

We will show that $Q$ is decreasing and convex and $R_\lambda$ is increasing and concave.

\begin{lem}
\label{lem-qdecreasing}
$Q(\lambda)$ is a decreasing function of $\lambda > 0$.
\end{lem}
\begin{proof}
The derivative of $Q$ may be computed as follows:
\begin{align}
Q'(\lambda) &= \partial_\lambda \{ F(\lambda,t(\lambda)) \}
    \nonumber \\
&= (\partial_\lambda F)(\lambda,t(\lambda)) + (\partial_e F)(\lambda,t(\lambda)) t'(\lambda)
    \nonumber \\
&= (\partial_\lambda F)(\lambda,t(\lambda))
    \nonumber \\
&= - \sum_{i=1}^{p} \frac{a_i \omega_i}{(a_i G(t(\lambda)) - \lambda)^2},
\end{align}
which is negative.
\end{proof}

\begin{prop}
\label{prop-qconvex}
$Q(\lambda)$ is a convex function of $\lambda > 0$.
\end{prop}
\begin{proof}
We first compute the derivative of $t(\lambda)$. We have 
\begin{align}
0 = (\partial_e F)(\lambda,t(\lambda)).
\end{align}
Differentiating with respect to $\lambda$, we find
\begin{align}
0 = (\partial_{\lambda e}^2 F)(\lambda,t(\lambda))
    + (\partial_{ee}^2 F)(\lambda,t(\lambda)) t'(\lambda),
\end{align}
and so
\begin{align}
t'(\lambda) &= \frac{-(\partial_{\lambda e}^2 F)(\lambda,t(\lambda))}
    {(\partial_{ee}^2 F)(\lambda,t(\lambda))}
    \nonumber \\
&= \frac{ -2 G'(t(\lambda)) 
        \sum_{i=1}^{p} \frac{a_i^2 \omega_i}{(a_i G(t(\lambda)) - \lambda)^3} }
    {G''(t(\lambda)) \sum_{i=1}^{p} \frac{a_i^2 \omega_i}{(a_i G(t(\lambda)) - \lambda)^2}
        - 2 G'(e)^2\sum_{i=1}^{p} \frac{a_i^3 \omega_i}{(a_i G(t(\lambda)) - \lambda)^3}
     }.
\end{align}

Now the second derivative of $Q$ is given by
\begin{align}
Q''(\lambda) = 2\sum_{i=1}^{p} \frac{a_i (a_i G'(t(\lambda)) t'(\lambda) -1) \omega_i}
    {(a_i G(t(\lambda)) - \lambda)^3}.
\end{align}

We will show that $Q''(\lambda) > 0$ for all $\lambda$. In fact, we will show the stronger result that each summand is positive, or equivalently, recalling that $a^* = \max_{1 \le i \le p} a_i$,
\begin{align}
a^* G'(t(\lambda)) t'(\lambda) \le 1.
\end{align}

To prove this, we observe that
\begin{align}
a^* G'(t(\lambda)) t'(\lambda)
&= \frac{ -2 G'(t(\lambda))^2 
        \sum_{i=1}^{p}\frac{a_i^3\omega_i}{(a_i G(t(\lambda)) - \lambda)^3} \frac{a^*}{a_i}}
    {G''(t(\lambda)) \sum_{i=1}^{p} \frac{a_i^2 \omega_i}{(a_i G(t(\lambda)) - \lambda)^2}
        - 2 G'(t(\lambda))^2 \sum_{i=1}^{p} 
        \frac{a_i^3 \omega_i}{(a_i G(t(\lambda)) - \lambda)^3}
    } 
    \nonumber \\
&= \frac{ -2 G'(t(\lambda))^2
        \sum_{i=1}^{p}\frac{a_i^3\omega_i}{(a_i G(t(\lambda)) - \lambda)^3} \frac{a^*}{a_i}}
    {G''(t(\lambda)) \sum_{i=1}^{p} \frac{a_i^3 \omega_i}{(a_i G(t(\lambda)) - \lambda)^3}
        \frac{a_i G(t(\lambda)) - \lambda} {a_i}
        - 2 G'(t(\lambda))^2 \sum_{i=1}^{p} 
        \frac{a_i^3 \omega_i}{(a_i G(t(\lambda)) - \lambda)^3}
    }
    \nonumber \\
&= \frac{ -2 G'(t(\lambda))^2
        \sum_{i=1}^{p} \frac{a_i^3 \omega_i}{(a_i G(t(\lambda)) - \lambda)^3} \frac{a^*}{a_i}}
    {\sum_{i=1}^{p} \frac{a_i^3 \omega_i}{(a_i G(t(\lambda)) - \lambda)^3}
        \left( G''(t(\lambda)) \frac{a_i G(t(\lambda)) - \lambda}{a_i}
            - 2 G'(t(\lambda))^2 \right)
    }
    \nonumber \\
&= \frac{\sum_{i=1}^{p} 
        \frac{a_i^3 \omega_i}{(\lambda - a_i G(t(\lambda)))^3} \frac{a^*}{a_i}}
    {\sum_{i=1}^{p} \frac{a_i^3 \omega_i}{(\lambda - a_i G(t(\lambda)))^3}
        \left( 1 - \frac{G''(t(\lambda))}{2 G'(t(\lambda))^2} 
                \frac{a_i G(t(\lambda)) - \lambda}{a_i} \right).
    }
\end{align}
To show that this is less than 1, it is enough to show that
\begin{align}
\frac{a^*}{a_i} \le 1 - \frac{G''(t(\lambda))}{2 G'(t(\lambda))^2} 
                \frac{a_i G(t(\lambda)) - \lambda}{a_i}
= 1 - \frac{G''(t(\lambda))}{2 G'(t(\lambda))^2} \left(G(t(\lambda)) - \frac{\lambda}{a_i} 
    \right),
\end{align}
or equivalently
\begin{align}
1 \le \frac{a_i}{a^*}
    + \frac{G''(t(\lambda))}{2 G'(t(\lambda))^2} \left(\frac{\lambda}{a^*} 
        - \frac{a_i}{a^*} G(t(\lambda))
    \right),
\end{align}

We will show that this inequality holds for any value of $a_i$ between $0$ and $a^*$. If $a_i = a^*$, then the right side becomes
\begin{align}
1 + \frac{G''(t(\lambda))}{2 G'(t(\lambda))^2} \left(\frac{\lambda}{a^*} 
        - G(t(\lambda))
    \right),
\end{align}
and since the term 
\begin{align}
\frac{G''(t(\lambda))}{2 G'(t(\lambda))^2} \left(\frac{\lambda}{a^*} 
        - G(t(\lambda))
    \right)
\end{align}
is positive (because $G$ is convex, $t(\lambda) \in J_\lambda$, and $a^* G(e) < \lambda$ on $J_\lambda$), the desired inequality is satisfied.

On the other hand, if $a_i = 0$, the right hand side becomes
\begin{align}
\frac{G''(t(\lambda))}{2 G'(t(\lambda))^2} \frac{\lambda}{a^*}
> \frac{G''(t(\lambda)) G(\lambda)}{2 G'(t(\lambda))^2}
\end{align}
where the inequality holds since $t(\lambda) \in J_\lambda$, and $a^* G(e) < \lambda$ for all $e \in J_\lambda$. By the Cauchy-Schwarz inequality,
\begin{align}
|G'(\lambda)| = \gamma\sum_{j=1}^{n} \left( \frac{b_j}{1 + \gamma b_j e} \right)^2  \pi_j
&=\gamma\sum_{j=1}^{n} \left( \frac{b_j}{1 + \gamma b_j e} \right)^{3/2} 
    \left( \frac{b_j}{1 + \gamma b_j e} \right)^{1/2} \pi_j
    \nonumber \\
& \le  \left[\gamma^2\sum_{j=1}^{n} \left( \frac{b_j}{1 + \gamma b_j e}
        \right)^{3}  \pi_j \right]^{1/2} 
    \left[\sum_{j=1}^{n} \frac{b_j}{1 + \gamma b_j e}  \pi_j \right]^{1/2}
    \nonumber \\
&= \sqrt{\frac{G''(e) G(e)}{2}},
\end{align}
or in other words,
\begin{align}
\frac{G''(e) G(e)}{2 G'(e)^2} \ge 1.
\end{align}
This gives the desired result.
\end{proof}

\begin{prop}
For all $\lambda > 0$, $R_\lambda(e)$ is an increasing and concave function of $e \in I_\lambda$.
\end{prop}
\begin{proof}
The first derivative of $R_\lambda$ is:
\begin{align}
R_\lambda'(e) = (\partial^2_{ee} F)(\lambda,e),
\end{align}
which as we've seen is positive. The second derivative of $R_\lambda$ is
\begin{align}
R_\lambda''(e) =& (\partial_{eee}^3 F)(\lambda,e)
    \nonumber \\
=& G^{(3)}(e)\sum_{i=1}^{p} \left(\frac{a_i}{a_i G(e) - \lambda}\right)^2 \omega_i
    - 6 G''(e) G'(e)\sum_{i=1}^{p} \left(\frac{a_i}{a_i G(e) - \lambda}\right)^3 \omega_i
    \nonumber \\
    &+ 6 G'(e)^3 \sum_{i=1}^{p} \left(\frac{a_i}{a_i G(e) - \lambda}\right)^4 \omega_i,
\end{align}
which is always negative, since $G^{(3)}(e) < 0$, $G''(e) > 0$, $G'(e) < 0$, and $a_i G(e) < \lambda$ for all $e \in I_\lambda$ and $1 \le i \le p$.
\end{proof}

\section{Algorithms}
\label{sec-algorithms}

In this section, we describe the algorithms for computing the SSDT $\lambda^*$ and for evaluating the Stieltjes transform $s(\lambda)$ and its derivative $s'(\lambda)$ at values $\lambda > \lambda^*$. By Lemmas \ref{lem-s2sbar} and \ref{lem-s2sbar-der}, $\s(\lambda)$ and $\s'(\lambda)$ may be easily found from $s(\lambda)$ and $s'(\lambda)$.

\subsection{Computation of the boundary $\lambda^*$}
\label{sec-edge}

In this section, we derive an algorithm for the computation of $\lambda^*$. First, we observe that when $\lambda > \lambda^*$, then as we have shown there are real roots of $F_\lambda(e)$ to the left and to the right of $t(\lambda)$; in particular, $F(\lambda,t(\lambda)) < 0$. On the other hand, if $\lambda < \lambda^*$, the function $F_\lambda(e)$ cannot have a real root, since that would imply that the Stieltjes transform is real inside the support of $\mu$. Consequently, $F(\lambda,t(\lambda)) > 0$. It follows that the SSDT $\lambda^*$ is the value at which $F(\lambda,t(\lambda)) = 0$;  that is, $\lambda^*$, is the unique root of $Q$ on $(0,\infty)$.

From Lemma \ref{lem-qdecreasing} and Proposition \ref{prop-qconvex}, $Q(\lambda)$ is decreasing and convex. With an efficient procedure for evaluating $Q(\lambda)$ and $Q'(\lambda)$, we can therefore use Newton's algorithm to find its root if we initialize the algorithm to the left of the root. In Section \ref{eval-q}, we detail how to evaluate $Q(\lambda)$ and $Q'(\lambda)$. As a preliminary step, we will need to compute the left endpoint of $I_\lambda$; we do this in \ref{sec-left}. The resulting algorithm for evaluating $\lambda^*$ is summarized in Algorithm \ref{alg:lambda_star}.

\begin{algorithm}[ht]
\caption{Computation of the left endpoint of $I_\lambda$.}
\label{alg:Ilambda_boundary}
\begin{algorithmic}[1]
\item {\bf Input:} Precision $\epsilon > 0$; parameter $\lambda > 0$

\item
{\bf Initialize:} $e > -1/(\gamma b^*)$

\item
{\bf Bisection:} $e \leftarrow (e -1/(\gamma b^*) )/2$, until $T_\lambda(e) > 0$

\item
{\bf Newton:} $e \leftarrow e - T_\lambda(e)/T_\lambda'(e)$, until $|T_\lambda(e)| < \epsilon$

\item 
{\bf Output:} $e_\lambda^* = e$

\end{algorithmic}
\end{algorithm}

\begin{algorithm}[ht]
\caption{Evaluation of $t(\lambda)$, $Q(\lambda)$ and $Q'(\lambda)$.}
\label{alg:tlambda}
\begin{algorithmic}[1]
\item {\bf Input:} Precision $\epsilon > 0$; parameter $\lambda > 0$

\item
{\bf Endpoint:} Compute $e_\lambda^*$ using Algorithm \ref{alg:Ilambda_boundary}

\item
{\bf Initialize:} $e > e_\lambda^*$

\item
{\bf Bisection:} $e \leftarrow (e + e_\lambda^* )/2$, until $R_\lambda(e) < 0$

\item
{\bf Newton:} $e \leftarrow e - R_\lambda(e)/R_\lambda'(e)$, until $|R_\lambda(e)| < \epsilon$

\item 
{\bf Output:} $t(\lambda) = e$,
$Q(\lambda) = F(\lambda,e)$, $Q'(\lambda) = (\partial_\lambda F)(\lambda,e)$

\end{algorithmic}
\end{algorithm}

\begin{algorithm}[ht]
\caption{Evaluation of $\lambda^*$.}
\label{alg:lambda_star}
\begin{algorithmic}[1]
\item {\bf Input:} Precision $\epsilon > 0$

\item
{\bf Initialize:} $\lambda > 0$

\item
{\bf Bisection:} $\lambda \leftarrow \lambda/2$, until $Q(\lambda) < 0$

\item
{\bf Newton:} $\lambda \leftarrow \lambda - Q(\lambda)/Q'(\lambda)$, until $|Q(\lambda)| < \epsilon$

\item 
{\bf Output:} $\lambda^* = \lambda$

\end{algorithmic}
\end{algorithm}

\subsubsection{Computation of the left endpoint of $I_\lambda$, $\lambda > 0$}
\label{sec-left}

We recall the definition of the interval $I_\lambda$:
\begin{align}
I_\lambda = \left\{ e \in J : G(e) < \frac{\lambda}{a^*} \right\},
\end{align}
where $J$ is the interval $J = \left\{ e : e > -1/(\gamma b^*) \right\}$. Let us denote by $e_\lambda^*$ the left endpoint of $I_\lambda$. Since $G(e)$ is a decreasing function of $e$ on $J$, $e_\lambda^*$ is the unique root of 
\begin{align}
T_\lambda(e) = G(e) - \frac{\lambda}{a^*}
\end{align}
on $J$. Since $T_\lambda$ is a decreasing, convex function of $e$ on $J$, we may find its root using Newton's algorithm, initialized to the left of the root. Such an initial value $e_0$ can be found by starting with any value $e$ in $J$, and performing bisection with $-1/(\gamma b^*)$, the left endpoint of $J$, until we arrive at a value where $T_\lambda(e_0) > 0$. Newton's algorithm is then performed with initial value $e_0$. We summarize the procedure in Algorithm \ref{alg:Ilambda_boundary}.

\subsubsection{Evaluation of $t(\lambda)$, $Q(\lambda)$, and $Q'(\lambda)$}
\label{eval-q}

Next, we show how to evaluate the functions $t(\lambda)$, $Q(\lambda)$, and $Q'(\lambda)$. $t(\lambda)$ is defined as the root of $R_\lambda(e)$ on $I_\lambda$. Since $R_\lambda(e)$ is an increasing, concave function, we may find its root using the Newton algorithm initialized to the left of the root, i.e.\ the region where $R_\lambda(e) < 0$. Such an initial value may be found by taking a starting point $e$ to the right of $e_\lambda^*$, and performing bisection on $e$ and $e_\lambda^*$ until we arrive at a value $e_0$ with $R_\lambda(e_0) < 0$. We can then perform Newton's algorithm on $R_\lambda$, initialized at $e_0$.  The procedure is summarized in Algorithm \ref{alg:tlambda}.

\subsection{Evaluation of $s(\lambda)$ and $s'(\lambda)$, $\lambda > \lambda^*$}
\label{sec-evaluate}

In this section, we present an algorithm for evaluating the Stieltjes transform $s(\lambda)$ of $\mu$, and its derivative $s'(\lambda)$, when $\lambda > \lambda^*$. This immediately provides a method for evaluating $\s(\lambda)$ and $\s'(\lambda)$, and the $D$-transform $D(\lambda)$ defined in Section \ref{sec-dtransform}.

As we showed in Section \ref{sec-fbehavior}, the function $F_\lambda(e)$ is convex on $I_\lambda$ and has two roots, both of which are negative. The root closest to $0$, i.e.\ the rightmost root, is $e(\lambda)$. Since $F_\lambda(0) > 0$ and $F_\lambda(e)$ is convex, this tells us that Newton's method, initialized at $e_0=0$, will converge to $e(\lambda)$. For brevity, we introduce the function $W(\lambda,e)$ defined by:
\begin{align}
W(\lambda,e) = \sum_{i=1}^{p} \frac{\omega_i}{a_i G(e) - \lambda}.
\end{align}
With this notation, we have:
\begin{align}
s(\lambda) = W(\lambda,e(\lambda))
\end{align}
and 
\begin{align}
s'(\lambda) = (\partial_\lambda W)(\lambda,e(\lambda)) 
    + (\partial_e W)(\lambda,e(\lambda)) e'(\lambda).
\end{align}

Note that from \eqref{eq12345}, we can evaluate $e'(\lambda)$:
\begin{align}
e'(\lambda) = \frac{-(\partial_\lambda F)(\lambda,e(\lambda))}
    {(\partial_e F)(\lambda,e(\lambda))}.
\end{align}

The method for evaluating $s(\lambda)$ and $s'(\lambda)$ is summarized in Algorithm \ref{alg:stieltjes}. 

\begin{algorithm}[h]
\caption{Evaluation of $s(\lambda)$ and $s'(\lambda)$.}
\label{alg:stieltjes}
\begin{algorithmic}[1]
\item {\bf Input:} Precision $\epsilon > 0$; parameter $\lambda > \lambda^*$

\item
{\bf Initialize:} $e = 0$

\item
{\bf Newton:} $e \leftarrow e - F_\lambda(e)/F_\lambda'(e)$, until $|F_\lambda(e)| < \epsilon$

\item 
{\bf Output:} $s(\lambda) = W(\lambda,e)$, 
    $s'(\lambda) =  (\partial_\lambda W)(\lambda,e) 
    -(\partial_e W)(\lambda,e)(\partial_\lambda F)(\lambda,e) / 
        (\partial_e F)(\lambda,e)$

\end{algorithmic}
\end{algorithm}

\section{Numerical results}
\label{sec-numerical}

In this section, we report on several experiments illustrating the behavior of the algorithms from Section \ref{sec-algorithms}. For the timings reported in Sections \ref{sec:scalability} and \ref{sec:compare}, we used an implementation written in MATLAB 2019b and run on a Dell Precision 5540 with 62.5 GB of RAM and an Intel Core i9 CPU. The MATLAB code is available at the following URL: \texttt{https://github.com/wleeb/MPBoundary}.

\subsection{Convergence}

Algorithms \ref{alg:Ilambda_boundary} -- \ref{alg:stieltjes} are all versions of the Newton root-finding method. In this section, we illustrate the quadratic convergence of these methods, as predicted from the theory reviewed in Section \ref{sec-newton}. In each experiment, we used parameters $p=512$, $n=1024$, and $\gamma = 1/2$. We generated the values $a_1,\dots,a_p$ and $b_1,\dots,b_n$ randomly from a $\text{Unif}(0,1)$ distribution, and assigned random probabilities $\omega_i$ and $\pi_j$ by drawing values from $\text{Unif}(0,1)$ and normalizing to sum to 1. For experiments in which a value $\lambda$ is specified, we also choose it at random.

In Tables \ref{table:Ilambda_boundary} -- \ref{table:stieltjes}, the first column displays the iteration number $m$, starting from the initial value and going until the root has been reached. The second column displays the function value at the $m^{th}$ iterate; the algorithm terminates when the function reaches machine precision $\epsilon$. We work in double precision, so $\epsilon \approx 10^{-16}$. The third column displays the relative error in the root itself, defined by:
\begin{align}
\text{err}(x_m,x) = \frac{x_m - x}{x}.
\end{align}

Table \ref{table:Ilambda_boundary} shows the results for Algorithm \ref{alg:Ilambda_boundary}, which computes the left endpoint $e_\lambda^*$ of $I_\lambda$. Table \ref{table:tlambda} shows the results for Algorithm \ref{alg:tlambda}, which evaluates $t(\lambda)$. Table \ref{table:lambda_star} shows the results for Algorithm \ref{alg:lambda_star}, which computes the boundary $\lambda^*$. Table \ref{table:stieltjes} shows the results for Algorithm \ref{alg:stieltjes}, which evaluates $e(\lambda)$.

\begin{rmk}
For each algorithm, we observe quadratic convergence close to the root, as expected. That is, on each iteration close to termination the number of correct digits roughly doubles, and the size of the objective roughly squares, until machine precision is reached.
\end{rmk}

\subsection{Scalability}
\label{sec:scalability}

In the next experiment, we compute timings for the computation of $\lambda^*$ and the evaluation of $s(\lambda)$. For increasing values of $n$, we set $p = n/2$ ($\gamma=1/2$). We generated the values $a_1,\dots,a_p$ and $b_1,\dots,b_n$ randomly from a $\text{Unif}(1,2)$ distribution, and assigned them random probabilities $\omega_i$ and $\pi_j$ by drawing values from $\text{Unif}(0,1)$ and normalizing to sum to 1. 

For each $n$, we then record the time in seconds required to compute $\lambda^*$, and the time in seconds required to compute $s(\lambda)$ on a grid of 100 equispaced values of $\lambda$ between $\lambda^* + 1$ and $\lambda^* + 10$. The reported timings are averaged over five runs of the experiment, and are displayed in Table \ref{table:timing}. It is apparent that the running times scale approximately linearly with $n$, as we would expect. In addition to linearly scaling with $n$, the magnitudes of the timings are quite encouraging; for example, it takes only about 4 seconds to compute $\lambda^*$ when $n$ is over two million and $p$ is over one million.


%

\subsection{Finite sample accuracy for $\lambda^*$}

The master equations \eqref{eq-e2s} -- \eqref{eq-gdef} from which we compute $\lambda^*$ are asymptotic and deterministic, holding almost surely in the limit as $k,l \to \infty$. In this experiment, we assess the finite sample accuracy of estimating the SSDT $\lambda^*$ from the operator norm of the random matrix $\N \N^T$. We consider how the estimate improves as $k$ and $l$ grow. We use a model where $p = n = 2$, and both $\nu$ and $\unu$ have point masses at $2$ and $3$, each with weight $1/2$; and set $\gamma = 1/2$. We generate matrices of size $k$-by-$l$, where $k$ grows and $l = k / \gamma = 2k$.

For each value of $k$, we draw such a $k$-by-$l$ random matrix $\N$ and compute the operator norm of $\N \N^T$. We compare this to the value of $\lambda^*$ computed using Algorithm \ref{alg:lambda_star}. In Table \ref{table:accuracy_edge}, we show the mean absolute error, averaged over $M=40000$ runs for each value of $k$. More precisely, if $\hat \lambda_j^*$ is the operator norm of $\N \N^T$ from trial $j=1,\dots,M$, we record
\begin{align}
\text{mean absolute error} 
= \frac{1}{M} \sum_{j=1}^{M} \frac{| \lambda^* - \hat\lambda_j^*|}{\lambda^*}.
\end{align}
We also record the average bias, defined as
\begin{align}
\text{mean bias} 
= \frac{1}{M} \sum_{j=1}^{M} \frac{ \lambda^* - \hat \lambda_j^*}{\lambda^*}.
\end{align}

In Figure \ref{fig-edge}, we plot the log error against $\log_2(k)$. The plot demonstrates a linear dependence. The slope of the line is estimated to be approximately $-0.68$.

\begin{rmk}
In \cite{johnstone2001distribution}, it is shown that for white noise the expected fluctations of the top eigenvalue of $\N \N^T$ are of the order $k^{-2/3}$, which implies the log-log plot would have a slope of $-2/3$. The observed slope of approximately $-0.68$ is a close match to this value.
\end{rmk}

\begin{rmk}
For all values of $k$, the bias is positive. In other words, the estimated values $\hat \lambda_k^*$ tend to underestimate $\lambda^*$.
\end{rmk}

\subsection{Finite sample accuracy of spiked model parameters}

In this experiment, we test the finite sample accuracy of parameter estimation in the spiked random matrix model. We consider $k$-by-$l$ random matrices of the form $\Y = \X + \N$, where $\X = \theta \u \v^T$ is a rank $1$ ``signal'' matrix with uniformly random singular vectors $\u$ and $\v$, and $\N = \A^{1/2} \G \B^{1/2}$ is a random Gaussian noise matrix with separable variance profile. We will denote by $\hat \lambda$

As we reviewed in Section \ref{sec-dtransform}, the top eigenvalue of $\Y \Y^T$ converges almost surely to a value $\lambda$ satisfying $\theta^2 = 1/D(\lambda)$, in the limit $k/l \to \gamma$. Furthermore, if $\hat \u$ and $\hat \v$ are the top singular vectors of $\Y$, then the absolute inner products $|\langle \u , \hat \u \rangle|$ and $|\langle \v , \hat \v \rangle|$ converge almost surely to $c \equiv s(\lambda_m) D(\lambda_m) / D'(\lambda_m)$ and $ \underline{c} \equiv \s(\lambda_m) D(\lambda_m) / D'(\lambda_m)$, respectively. We remind the reader that we define the $D$-transform by $D(\lambda) = \lambda s(\lambda) \s(\lambda)$.

We test the accuracy of these formulas for finite and increasing values of $k$ and $l$. Using Algorithm \ref{alg:stieltjes} for evaluating $s(\lambda)$ and $s'(\lambda)$, we can easily evaluate $D(\lambda)$.  Proposition \ref{prop-dtransform} shows that for a specified parameter $\theta$, Newton's root-finding algorithm may be used to solve for the asymptotic $\lambda = 1 / D^{-1}(\theta^2)$. The asymptotic values $c$ and $\underline{c}$ are then evaluated from their respective formulas.

We compare these asymptotic values to the top eigenvalue of $\N \N^T$ and the cosines between the singular vectors of $\X$ and $\Y$, for randomly generated data. We again use a model where $p = n = 2$, and both $\nu$ and $\unu$ have point masses at $2$ and $3$, each with weight $1/2$; and set $\gamma = 1/2$. We generate matrices of size $k$-by-$l$, where $k$ grows and $l = k / \gamma = 2k$. We generate a signal matrix $\X = \theta \u \v^T$ where $\u$ and $\v$ are uniformly random unit vectors in $\R^k$ and $\R^l$, respectively; and $\theta = \sqrt{1 / D(\lambda^*) + 20}$, ensuring a detectable signal.

For each value of $k$, we draw such a $k$-by-$l$ random matrix $\Y = \X + \N$ and compute the operator norm of $\Y \Y^T$. We compare this to the asymptotic value of $\lambda$. We also compare the values $|\langle \u ,\hat \u \rangle|$ and $|\langle \v , \hat \v \rangle|$ to $c$ and $\underline{c}$, respectively. In Table \ref{table:accuracy_spike}, we show the mean absolute errors of these estimates, averaged over $M=40000$ runs for each value of $k$. More precisely, if $\hat \lambda_j$ is the operator norm of $\Y \Y^T$ from trial $j=1,\dots,M$, and $\hat c_j$ and $\hat {\underline c}_j$ are the left and right cosines, respectively, we record the mean absolute errors:
\begin{align}
\frac{1}{M} \sum_{j=1}^{M} \frac{| \lambda - \hat\lambda_j|}{\lambda},
\quad
\frac{1}{M} \sum_{j=1}^{M} \frac{| c - \hat c_j|}{c},
\quad
\frac{1}{M} \sum_{j=1}^{M} \frac{| \underline{c} - \hat {\underline c}_j|}{\underline{c}}.
\end{align}

In Figure \ref{fig-spike}, we plot the log errors against $\log_2(k)$. The plots all demonstrate a linear dependence. The slope of each line is estimated to be approximately $-0.50$.

\begin{rmk}
In \cite{benaych2012singular}, it is shown that the expected fluctations of the top eigenvalue of $\Y \Y^T$ are of the order $k^{-1/2}$, which implies the log-log plot would have a slope of $-1/2$. In \cite{bao2018singular}, it is shown that in the case of white noise the expected fluctations of the cosines are also of order $k^{-1/2}$, also leading to a slope of $-1/2$. Our observed slopes closely match these values.
\end{rmk}

\subsection{One-sided weights}
\label{sec:compare}

While Algorithm \ref{alg:lambda_star} computes the SSDT for an arbitrary separable variance profile, a special case of this problem that arises in certain applications is when $\B = \I_n$; that is, the variance profile has one-sided weights. In \cite{dobriban2017deterministic}, it is shown that the SSDT may be evaluated by finding the minimizer $v^*$ of the function 
\begin{align}
z(v) = \frac{-1}{v} + \gamma\sum_{i=1}^{p} \frac{a_i \omega_i}{1 + a_i v}
\end{align}
on the interval $(-1/a^*,0)$, and then setting $\lambda_* = z(v^*)$. The minimizer $v^*$ is the unique root of the monotonic function
\begin{align}
\label{eq-zprime}
z'(v) = \frac{1}{v^2} - \gamma \sum_{i=1}^{p} \frac{a_i^2 \omega_i}{(1 + a_i v)^2}
\end{align}
on the interval $(-1/a^*,0)$. 

When viable, this approach has obvious advantages over Algorithm \ref{alg:lambda_star}, namely that it finds the root of a function which can be evaluated in closed form rather than by nested applications of Newton's method. Properly applied, it should be substantially faster than Algorithm \ref{alg:lambda_star}. However, the paper \cite{dobriban2017deterministic} does not analyze the behavior of the function \eqref{eq-zprime}, beyond observing that it is monotonic.

We compare the method from \cite{dobriban2017deterministic} for evaluating $\lambda_*$ based on the minimizer of $z(v)$ to Algorithm \ref{alg:lambda_star}. To find the root of $z'(v)$, we employ bisection until the error is approximately square root of machine epsilon, and then use Newton's method to achieve full accuracy. In our experiment, we set $\gamma = 1/2$, and  for each value of $p$ we generate the $a_i$'s uniformly randomly on $[0,1]$ and take uniform weights $\omega_i = 1/p$. For each value of $p$, we solve the problem using each method for $10$ runs, and average the timings of all the runs; the timings are presented in Table \ref{table:compare}.

The results of the experiment are extremely encouraging. They suggest that for one-sided weights, working with the function \eqref{eq-zprime} yields a substantially faster computation of the SSDT than the general master equations  \eqref{eq-e2s} -- \eqref{eq-gdef}. However, a more detailed analysis of the function \eqref{eq-zprime} must be carried out to justify such an approach and show that a fast algorithm is viable for general $a_i$ and $\omega_i$; this is beyond the scope of the current work.

\section{Conclusion}

We have introduced an algorithm for rapidly computing the spectral signal detection threshold $\lambda^* $ in signal-plus-noise random matrix models. We have considered a class of random noise matrices with separable variance structure, which arise often in applications. Our algorithm is based on an implicit characterization of $\lambda^*$ derived from the master equations for the Stieltjes transform $s(\lambda)$. Several nested applications of Newton's method are applied to evaluate $\lambda^*$ and auxiliary parameters. Fast algorithms are also introduced to evaluate $s(\lambda)$ and $s'(\lambda)$, the Stieltjes transform and its derivative, at real values $\lambda > \lambda^*$. We have demonstrated the rapid convergence of these methods and their linear scaling in numerical tests. We have also shown the effects of random fluctations for increasing values of $p$ and $n$.

\section*{Acknowledgements}

I acknowledge support from NSF BIGDATA award IIS 1837992 and BSF award 2018230. I thank Edgar Dobriban for helpful discussions and for pointing out the method from \cite{dobriban2017deterministic}. I also thank the reviewers for their helpful comments on the manuscript.

\bibliographystyle{plain}
\bibliography{refs_mp}

\newpage

%
%
\begin{table} 
\centering
\caption{Error after each iterate in evaluation of $e_\lambda^*$.}
\label{table:Ilambda_boundary}
\begin{tabular}{|c| c  | c |}  
\hline  
 $m$ &  $T_\lambda(e_m)$ & err($e_m,e_\lambda^*$) \\ 
\hline  
1 & 2.08e-01 & 4.26e-02  \\ 
2 & 4.34e-02 & 1.11e-02  \\ 
3 & 2.55e-03 & 6.90e-04  \\ 
4 & 9.66e-06 & 2.63e-06  \\ 
5 & 1.40e-10 & 3.80e-11  \\ 
6 & -2.75e-16 & -1.23e-16  \\ 
\hline  
\end{tabular}  
\end{table}

%
%
\begin{table}
\centering
\caption{Error after each iterate in evaluation of $t(\lambda)$.}
\label{table:tlambda}
\begin{tabular}{|l| c  | c |}  
\hline  
$m$& $R_\lambda(e_m)$ & err($e_m,t(\lambda)$) \\ 
\hline  
1 & -3.77e-01 & 3.71e-02  \\ 
2 & -6.89e-02 & 8.34e-03  \\ 
3 & -3.40e-03 & 4.34e-04  \\ 
4 & -9.24e-06 & 1.18e-06  \\ 
5 & -6.84e-11 & 8.73e-12  \\ 
6 & -5.55e-15 & 7.85e-16  \\ 
\hline  
\end{tabular}  
\end{table}

%
%
\begin{table} 
\centering
\caption{Error after each iterate in evaluation of $\lambda^*$.}
\label{table:lambda_star}
\begin{tabular}{|l| c  | c |}  
\hline  
$m$ &  $Q(\lambda_m)$ & err($\lambda_m,\lambda^*$) \\ 
\hline  
1 & 1.25e+00 & -3.10e-01  \\ 
2 & 3.10e-01 & -1.03e-01  \\ 
3 & 3.12e-02 & -1.15e-02  \\ 
4 & 3.92e-04 & -1.46e-04  \\ 
5 & 6.34e-08 & -2.37e-08  \\ 
6 & 2.00e-15 & -2.11e-16  \\ 
\hline  
\end{tabular}  
\end{table}

%
%
\begin{table} 
\centering
\caption{Error after each iterate in evaluation of $e(\lambda)$.}
\label{table:stieltjes}
\begin{tabular}{|l| c  | c |}  
\hline  
 $m$ &  $F_\lambda(e_m$) & err($e_m,e(\lambda)$) \\ 
\hline  
1 & 4.71e-01 & -1.00e+00  \\ 
2 & 8.86e-03 & -1.93e-02  \\ 
3 & 6.42e-06 & -1.40e-05  \\ 
4 & 3.43e-12 & -7.50e-12  \\ 
5 & 4.44e-16 & -1.10e-15  \\ 
\hline  
\end{tabular}  
\end{table}

%
%
\begin{table}
\centering
\caption{Timings in seconds for evaluating $\lambda^*$ and $s(\lambda)$ at 100 values of $\lambda$.}
\label{table:timing}
\begin{tabular}{|c| c  | c |}  
\hline  
 $\log_2(n)$ &  Timing, $\lambda_*$ & Timing, $s(\lambda)$ \\ 
\hline  
   18 & 4.30e-01 & 1.93e-01  \\ 
   19 & 8.67e-01 & 4.27e-01  \\ 
   20 & 1.89e+00 & 1.66e+00  \\ 
   21 & 3.95e+00 & 4.20e+00  \\ 
   22 & 8.41e+00 & 1.02e+01  \\ 
   23 & 1.70e+01 & 2.09e+01  \\ 
   24 & 3.43e+01 & 4.19e+01  \\ 
\hline  
\end{tabular}  
\end{table}

%

%
%
\begin{table} 
\centering
\caption{Errors and bias in estimating $\lambda^*$.}
\label{table:accuracy_edge}
\begin{tabular}{|c| c  | c |}  
\hline  
 $\log_2(k)$ & Error & Bias \\ 
\hline  
    5 & 8.73e-02 & 7.72e-02   \\ 
    6 & 5.41e-02 & 4.75e-02   \\ 
    7 & 3.39e-02 & 2.93e-02   \\ 
    8 & 2.12e-02 & 1.82e-02   \\ 
    9 & 1.32e-02 & 1.13e-02   \\ 
   10 & 8.20e-03 & 6.97e-03   \\ 
   11 & 5.12e-03 & 4.32e-03   \\ 
\hline  
\end{tabular}  
\end{table}  
%

%

%
%
\begin{table} 
\centering
\caption{Errors in estimating $\lambda$, $c$, and $\underline{c}$.}
\label{table:accuracy_spike}
\begin{tabular}{|c| c  | c | c|}  
\hline  
 $\log_2(k)$ &  sing.\ val. & left cos. & right cos.  \\ 
\hline  
    5 & 7.68e-02 & 1.80e-02 & 2.26e-02 \\ 
    6 & 5.43e-02 & 1.27e-02 & 1.59e-02 \\ 
    7 & 3.85e-02 & 9.10e-03 & 1.13e-02 \\ 
    8 & 2.74e-02 & 6.38e-03 & 7.98e-03 \\ 
    9 & 1.92e-02 & 4.53e-03 & 5.61e-03 \\ 
   10 & 1.36e-02 & 3.20e-03 & 3.99e-03 \\ 
   11 & 9.58e-03 & 2.25e-03 & 2.81e-03 \\ 
\hline  
\end{tabular}  
\end{table}  
%

%
%
\begin{table}
\centering
\caption{Timings in seconds of Algorithm \ref{alg:lambda_star} and the solution via minimizing $z(v)$.}
\label{table:compare}
\begin{tabular}{|c| c  | c | }  
\hline  
 $\log_2(p)$ &  Alg.\ \ref{alg:lambda_star} & Min.\ $z(v)$  \\ 
\hline  
    18 & 2.67e-01 & 2.38e-02   \\ 
    19 & 5.29e-01 & 4.63e-02   \\ 
    20 & 1.14e+00 & 1.30e-01   \\ 
    21 & 2.43e+00 & 3.20e-01   \\ 
    22 & 5.37e+00 & 6.94e-01   \\ 
    23 & 1.09e+01 & 1.38e+00   \\ 
    24 & 2.18e+01 & 2.73e+00   \\ 
\hline  
\end{tabular}  
\end{table}

\newpage

%
%
\begin{figure}[t]
\center
\includegraphics[scale=.5]{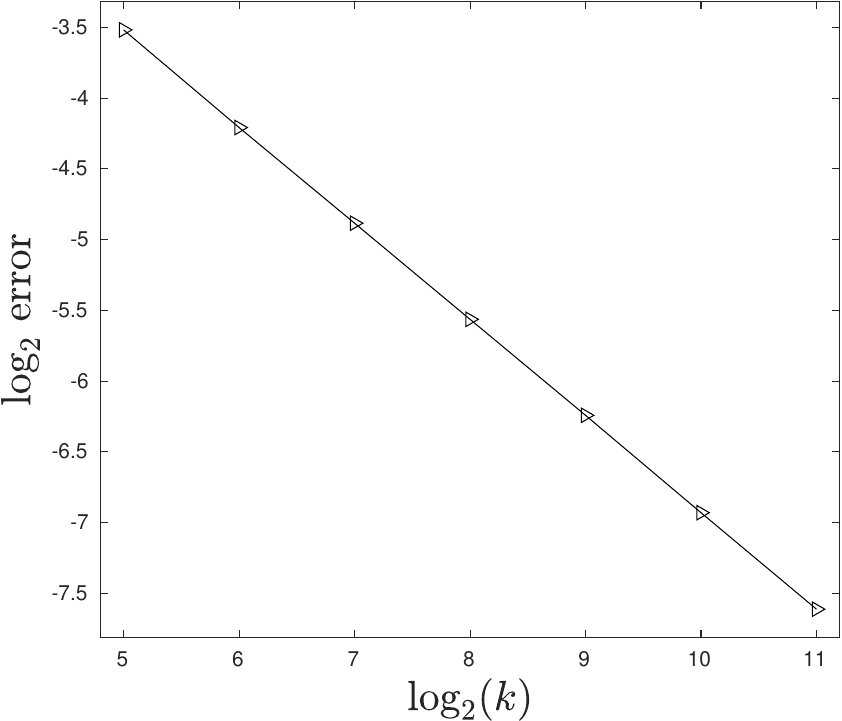}
\caption{Log errors for estimating $\lambda^*$.}
\label{fig-edge}
\end{figure}

%
%
\begin{figure}[h]
\center
\includegraphics[scale=.5]{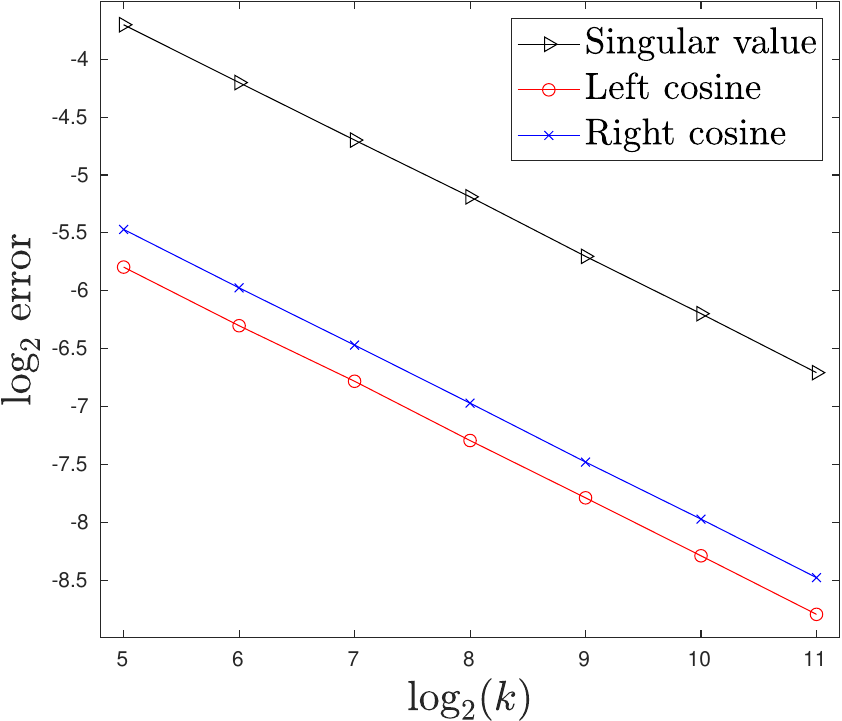}
\caption{Log errors for estimating $\lambda$, $c$, and $\underline{c}$.}
\label{fig-spike}
\end{figure}

\end{document}